\documentclass[aip,jcp,10pt,nofootinbib,twocolumn]{revtex4-1}
\usepackage{amsmath,amssymb,amsfonts,amsthm}
\usepackage[utf8]{inputenc}
\usepackage{graphicx}
\usepackage[caption=false]{subfig}
\usepackage[pdftex,bookmarks=false,colorlinks=true,linkcolor=blue,
citecolor=blue,filecolor=black,urlcolor=blue]{hyperref}

\providecommand{\ud}{\,\mathrm{d}}
\providecommand{\abs}[1]{\left\lvert#1\right\rvert}
\providecommand{\vect}[1]{{\mathbf{#1}}}
\providecommand{\R}{\mathbb{R}}

\providecommand{\calA}{{\mathcal{A}}}
\providecommand{\calB}{{\mathcal{B}}}
\providecommand{\calP}{{\mathcal{P}}}
\providecommand{\SCE}{{\mathrm{SCE}}}
\providecommand{\ee}{\mathrm{ee}}

\newtheorem{theorem}{Theorem}[section]

\newtheorem{definition}[theorem]{Definition}
\newtheorem{corollary}[theorem]{Corollary}

\begin{document}

\title{N-density representability and the optimal transport limit of the Hohenberg-Kohn functional}

\author{Gero Friesecke}
\affiliation{Mathematics Department, Technische Universit\"at M\"unchen, Garching bei M\"unchen, Germany}

\author{Christian B. Mendl}
\affiliation{Mathematics Department, Technische Universit\"at M\"unchen, Garching bei M\"unchen, Germany}

\author{Brendan Pass}
\affiliation{Department of Mathematical and Statistical Sciences, University of Alberta, Canada}

\author{Codina Cotar}
\affiliation{Department of Statistical Science, University College London, London, UK}

\author{Claudia Kl\"uppelberg}
\affiliation{Mathematics Department, Technische Universit\"at M\"unchen, Garching bei M\"unchen, Germany}

\begin{abstract}
We derive and analyze a hierarchy of approximations to the strongly correlated limit of the Hohenberg-Kohn functional. These ``density representability approximations'' are
obtained by first noting that in the strongly correlated limit, $N$-representability of the pair density reduces to the requirement that the pair density must come from
a symmetric $N$-point density. One then relaxes this requirement to the existence of a representing symmetric $k$-point density with $k<N$. The approximate energy can be computed by
simulating a fictitious $k$-electron system. We investigate the approximations by deriving analytically exact results for a $2$-site model problem, and by incorporating them into a self-consistent Kohn-Sham calculation for small atoms. We find that the low order representability conditions already capture the main part of the correlations.
\end{abstract}

\date{\today}

\maketitle

\section{Introduction}
\label{sec:Intro}

Kohn-Sham (KS) density functional theory (DFT) is currently the most widely used ab initio electronic structure model 
which is applicable to large and complex systems ranging from condensed matter over surfaces to nanoclusters and
biomolecules \cite{}. With the advent of linear scaling algorithms, the key factor limiting the accuracy of predictions
is the choice of underlying exchange-correlation functionals \cite{PerdewBurkeErnzerhof1996}. These functionals model the correlation structure of the system as a universal
functional of its underlying one-body density. While highly successful in many instances, these functionals exhibit known failures, both specific, 
like incorrect filling order and lack of binding of certain transition metal atoms \cite{Y00} or doubtful equilibrium geometries
of Carbon clusters \cite{MEF1996}, and general, like Van der Waals forces not being predicted.

It has long been recognized that important insight can be gained by studying the asymptotic relationship between correlation structure and one-body density in scaling limits \cite{Dirac, KohnSham1965, Becke1988}. In this paper we focus on the strongly correlated limit of the exact Hohenberg-Kohn functional first investigated in Ref.~\onlinecite{Seidl1999, SPL1999, SGS2007} in which electron repulsion dominates over kinetic energy (yielding a natural counterpart to the Kohn-Sham kinetic energy functional~\cite{KohnSham1965}).

The resulting limit, which can be interpreted as an optimal transport problem \cite{CFK1, BPG2012}, is still unwieldy from a computational point of view, since it requires 
the computation of the full $N$-point density of an $N$-electron system, a function on $\R^{3N}$. Here we give a simple but we believe fruitful
reformulation as a minimization problem over 2-point densities subject to a representability constraint. This is similar to the well known formulation of the full quantum $N$-body problem via representable 2-point density matrices \cite{Coleman2000, Mazziotti2007}, but an important difference is that here it is only required that the 2-point density
arises from a symmetric $N$-point density rather than from an antisymmetric, spin-dependent $N$-particle wavefunction. We therefore speak of \emph{$N$-density representability}. 
Density representability no longer mirrors the fermionic nature of electrons, reflecting the fact that a ``semi-classical'' limit has been taken of the Hohenberg-Kohn functional.

We then establish a natural hierarchy of necessary representability conditions, and investigate the accuracy of the resulting reduced models as compared to the full
strongly correlated limit. We focus on two test cases: first, a simple but illuminating 2-site, $N$-particle model in which all representability conditions can be computed explicitly; and second, ab initio as well as self-consistent densities for the atoms He, Li, Be. A tentative conclusion is that the low order representability conditions already capture the main part of the correlations, at significantly reduced computational cost.

\section{Strongly correlated limit of the Hohenberg-Kohn functional}
\label{sec:StrongCorrHK}

The following counterpart to the Kohn-Sham kinetic energy functional was introduced in Ref.~\onlinecite{Seidl1999, SPL1999, SGS2007}: 
\begin{equation} \label{eq:limFHK}
	V_{\ee}^{\SCE}[\rho] = \inf_{\Psi\mapsto\rho}\Bigl\langle \Psi | \widehat{V}_{\ee} | \Psi\Bigr\rangle
\end{equation}
Here the minimization is over electronic wavefunctions $\Psi = \Psi(\vect{x}_1,s_1,\dots,\vect{x}_N,s_N)$ which depend on $N$ space and spin coordinates and belong to the usual space $\calA_N$ of square-integrable antisymmetric normalized wavefunctions with square-integrable gradient, and the notation $\Psi\mapsto\rho$ means that $\Psi$ has single-particle density $\rho$. The acronym SCE stands for strictly correlated electrons \cite{SGS2007}. 
Here we briefly review known theoretical properties of $V_{\ee}^{\SCE}$ and previous approaches to compute it numerically.

Alternative constructions which plausibly yield the same functional are:\\
i) semiclassical limit of the Hohenberg-Kohn functional: 
\begin{equation}
	\label{eq:alternative1}
	\lim_{\hbar\to 0} \min_{\Psi\mapsto\rho, \, \Psi\in\calA_N} \Bigl\langle \Psi | \hbar^2 \, \widehat{T} + \widehat{V}_{\ee} | \Psi \Bigr\rangle,
\end{equation}
ii) minimization over spinless bosonic wavefunctions $\Phi$:
\begin{equation}
	\label{eq:alternative2}
	\inf_{\Phi\mapsto\rho, \, \Phi\in\calB_N} \Bigl\langle\Phi | \widehat{V}_{\ee} | \Phi \Bigr\rangle,
\end{equation}
iii) minimization over $N$-point probability measures:
\begin{equation}
	\label{eq:alternative3}
	\min_{\rho_N\mapsto\rho, \, \rho_N\in\calP_N^{\text{sym}}} \int_{\R^{3N}} V_{\ee} \, \rho_N.
\end{equation}
Here $\calB_N$ denotes the analogue of the space $\calA_N$ for spinless symmetric (bosonic) wavefunctions, $\calP_N^{\text{sym}}$ stands for the set of symmetric probability 
measures on $\R^{3N}$, and $\widehat{V}_{\ee}$ is the multiplication operator with the interaction potential
\begin{equation}
\label{eq:Vee}
	V_{\ee}(\vect{x}_1,\dots,\vect{x}_N) = \sum_{1 \le i < j \le N} v_{\ee}(\vect{x}_i,\vect{x}_j),
\end{equation}
where $v_{\ee}(\vect{x},\vect{y}) = \abs{\vect{x}-\vect{y}}^{-1}$. 
Formulae \eqref{eq:alternative1}, \eqref{eq:alternative2} appear in Ref.~\onlinecite{SGS2007}, and \eqref{eq:alternative1}, \eqref{eq:alternative3} are implicit in Ref.~\onlinecite{Seidl1999}.

The minimum value in \eqref{eq:alternative1} with $\hbar=1$ is the exact Hohenberg-Kohn functional $F^{\mathrm{HK}}$ in the Levy-Lieb constrained search formulation;
so expression \eqref{eq:alternative1} is the semiclassical limit of the Hohenberg-Kohn functional. 
Note that the Kohn-Sham kinetic energy functional $T^{\mathrm{KS}}$ is obtained from $F^{\mathrm{HK}}$ by instead retaining the kinetic
energy operator $\widehat{T}$ and neglecting the interaction term $\widehat{V}_{\ee}$. Expression \eqref{eq:alternative2} is related to \eqref{eq:limFHK} by neglecting antisymmetry and spin, and to expression \eqref{eq:alternative3} by first noting that $\langle\Phi|\widehat{V}_{\ee}|\Phi\rangle=\int V_{\ee} \, \abs{\Phi}^2$ and then replacing squares of spinless symmetric wavefunctions by their mathematical ``closure'', symmetric probability measures. 

Equality between the four expressions \eqref{eq:limFHK} -- \eqref{eq:alternative3} was conjectured in Ref.~\onlinecite{SGS2007} and has recently been justified mathematically \cite{CFK2}. 

As noticed in Ref.~\onlinecite{CFK1, BPG2012}, the last expression, \eqref{eq:alternative3}, has the form of an \emph{optimal transport problem}. In the standard setting of such problems\cite{Villani2008} originating from economics, one has $N=2$, $\rho_2(\vect{x},\vect{y})$ corresponds to the amount of ``mass'' 
transported from $\vect{x}$ to $\vect{y}$, $V_{\ee}(\vect{x},\vect{y})$ is the ``cost'' of this transport, the one-body densities of $\vect{x}$ and $\vect{y}$ would be different from each other
but prescribed a priori, i.e., $\int \rho_2(\vect{x},\vect{y}) \ud \vect{y} = \rho^A(\vect{x})$ and $\int \rho_2(\vect{x},\vect{y}) \ud \vect{x} = \rho^B(\vect{y})$, and minimization of $\int V_{\ee} \, \rho_2$ amounts to
finding the most economical way of transporting the pile of mass $\rho^A$ to $\rho^B$.
In economics, the cost would typically increase rather than decrease with distance, prototypical
examples being $\abs{\vect{x}-\vect{y}}$ or $\abs{\vect{x}-\vect{y}}^2$.

An interesting feature of minimizers is that they typically concentrate on lower-dimensional sets (see Fig.~\ref{fig:rho2marginal}). For $N = 2$, these sets have the form $\vect{y} = \vect{T}(\vect{x})$. Physically, this reflects the fact that
given the position of the first electron, the position of the second electron becomes deterministic in the strongly correlated limit \eqref{eq:limFHK}. When $\rho$ is radially
symmetric, $\vect{T}$ is known explicitly in terms of the inverse of the radial distribution function $R\mapsto 4\pi \int_0^R r^2\rho(r) \ud r$.\cite{SPL1999, Seidl1999, SGS2007, CFK1, BPG2012}

\begin{figure}[!ht]
\centering
\includegraphics[width=\columnwidth]{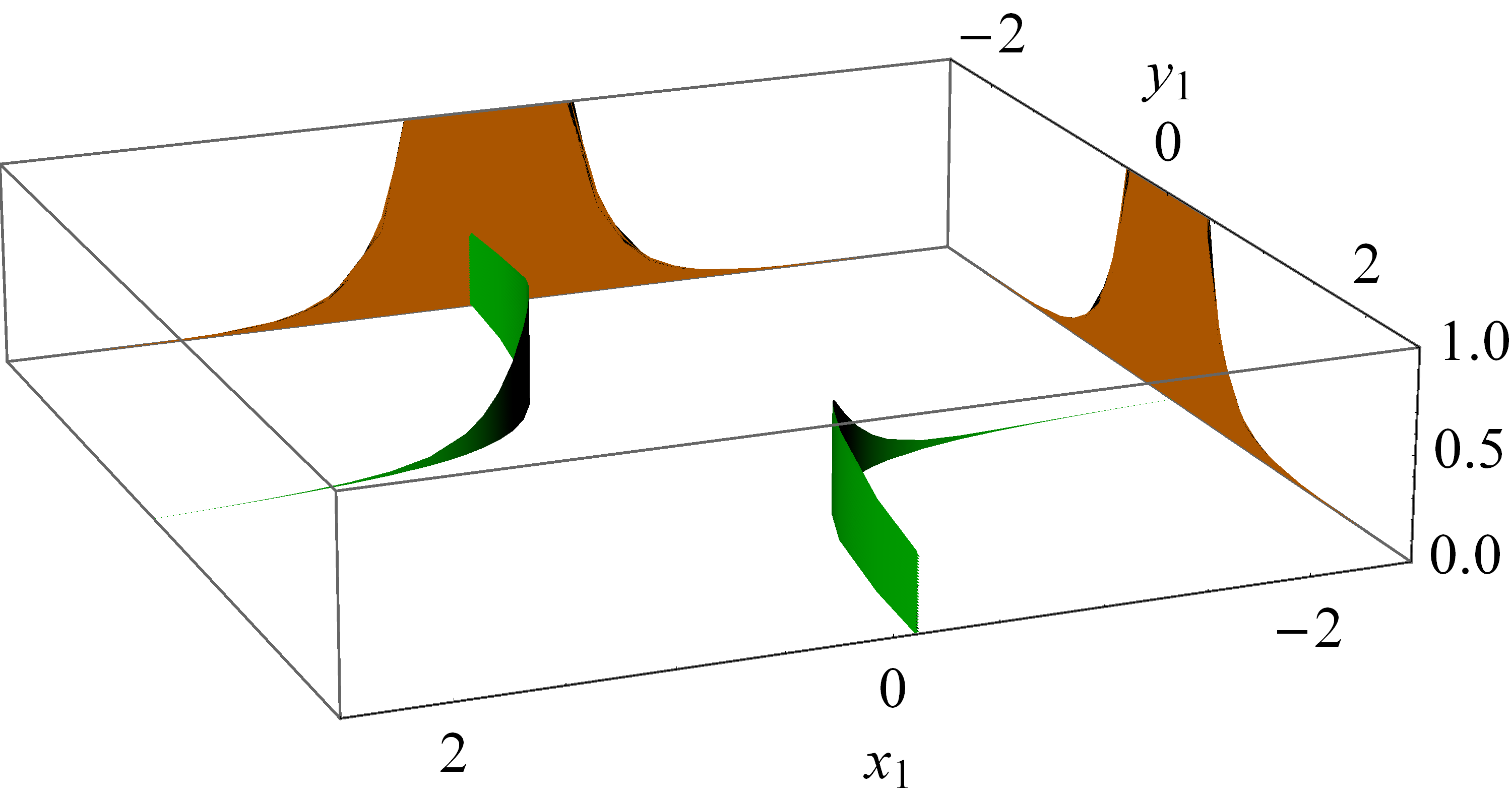}
\caption{Optimal pair density $\rho_2$ of Eq.~\eqref{eq:alternative3} evaluated on $(x_1,0,0,y_1,0,0)$ (green) for one-body ground-state density $\rho$ of helium (brown). The height of the green surface $(1+ \abs{\nabla \vect{T}(\vect{x})}^2)^{-1/2} \rho(\vect{x})$ indicates the prefactor of the Hausdorff measure on the set $\vect{y} = \vect{T}(\vect{x})$.}
\label{fig:rho2marginal}
\end{figure}

The minimization problem \eqref{eq:alternative3}, and indeed any optimal transport problem, has two alternative formulations. To obtain the so-called Monge formulation, one makes the ansatz
\begin{multline}
	\label{eq:ansatz}
	\rho_N(\vect{x}_1,\dots,\vect{x}_N) \\
	= \frac{\rho(\vect{x}_1)}{N} \delta(\vect{x}_2-\vect{T}_2(\vect{x}_1)) \cdots \delta(\vect{x}_N - \vect{T}_N(\vect{x}_1))
\end{multline}
for transport maps or \emph{co-motion functions} $\vect{T}_i: \, \R^3\to\R^3$ which preserves the one-body density $\rho$, that is to say $\int_{T(A)}\rho = \int_A\rho$ for general subsets $A\subset\R^3$
(see Ref.~\onlinecite{SPL1999, SGS2007, CFK1} for physical and mathematical justifications). In fact, the ansatz \eqref{eq:ansatz} is not in general symmetric, so strictly speaking
we should minimize over the symmetrizations of measures of form \eqref{eq:ansatz}, but dropping the symmetrization does not alter the minimum value in \eqref{eq:alternative3}. Or one passes to the so-called Kantorovich dual formulation
\begin{equation}
	\label{eq:Dual}
	V^{\SCE}_{\ee}[\rho] = \sup_{u \, : \R^3\to\R, \, \sum\limits_i u(\vect{x}_i) \le V_{\ee}(\vect{x}_1,\dots,\vect{x}_N)} \int \rho \, u
\end{equation}
(see Ref.~\onlinecite{BPG2012} for a mathematical justification and Ref.~\onlinecite{ML2013} for a numerical scheme). 

The Monge formulation amounts to a spectacular dimension reduction, in that the unknowns are $N$ maps on $\R^3$ instead of one function on $\R^{3 N}$. 
Thus, when discretizing $\R^3$ by $K$ gridpoints one has $K\cdot 3N$ instead of $K^N$ computational degrees of freedom.
However, for $N > 2$ it is not clear if the (symmetrized) Monge formalism captures all minimizers of \eqref{eq:alternative3}. See Section~\ref{sec:MongeCounterexample} for a counterexample when the Coulomb repulsion is replaced by a repulsive harmonic interaction. Moreover, previous numerical (and analytical) computation of $V_{\ee}^{\SCE}$ using \eqref{eq:ansatz} is hitherto restricted to spherically symmetric densities or 1D systems. This is because one has to deal with the infinite-dimensional
nonlinear constraint that the $\vect{T}_i$ preserve $\rho$, and because the $\vect{T}_i$ are expected to jump along surfaces; in the radial case, the surfaces are believed to be concentric spheres\cite{SGS2007}.



The Kantorovich dual formulation, which has been successfully applied to non-spherical problems \cite{ML2013}, cures the high storage complexity of the original formulation
\eqref{eq:alternative3}, but the inequality constraint in \eqref{eq:Dual} that needs to be satisfied by $u$ is still high-dimensional.

It is then of interest to explore alternative ways of reducing the dimensionality of \eqref{eq:limFHK}. The remainder of this paper is devoted to developing such an alternative approach, based on minimization over 2-point densities satisfying representability constraints.

\section{N-density representability and reduced density models}
\label{sec:NdensityRepr}

We now derive a simple but we believe fruitful reformulation of the minimization problem \eqref{eq:limFHK}. 
We begin by formalizing the notion of reduced densities.
Throughout this section it is useful to work with the convention that all densities and reduced densities 
integrate to $1$. We denote 
$k$-point reduced densities with this normalization by $p_k$, to distinguish them from the customary $k$-point 
reduced densities $\rho_k$ which integrate to the number of $k$-tuples in the system. Thus, given a symmetric $N$-point 
probability density $p_N$ on $\R^{3N}$, $N\ge 2$, we define the associated one- and two-point reduced densities 
(known in probability theory under the name marginal densities) by 
\begin{align}
	p_1(\vect{x}_1) &= \int_{\R^{3(N-1)}} p_N(\vect{x}_1,\dots,\vect{x}_N) \, \ud\vect{x}_2 \cdots \ud\vect{x}_N, \label{eq:onebodydens} \\
	p_2(\vect{x}_1,\vect{x}_2) &= \int_{\R^{3(N-2)}} p_N(\vect{x}_1,\dots,\vect{x}_N) \, \ud\vect{x}_3 \cdots \ud\vect{x}_N. \label{eq:twobodydens} 
\end{align}
In particular, $p_1$ is related to the customary one-body density $\rho$ by the formula $p_1=\rho/N$. 

Typical $N$-point densities occurring in the SCE limit concentrate on lower dimensional
subsets (see Figure~\ref{fig:rho2marginal}). Mathematically this does not pose real difficulties (and is a higher-dimensional analogue of the familiar
charge distributions on surfaces in electrostatics). It just means that
these densities should properly be regarded as probability measures on $\R^{3N}$, not functions, i.e., 
they are not specified by pointwise values but by their integrals over sets. In the more general setting of probability measures, \eqref{eq:onebodydens}, \eqref{eq:twobodydens} have to be replaced by 
\begin{align}
	\label{eq:marginals1}
    \int_{A} \ud p_1 &= \int_{A\times\R^{3(N-1)}} \ud p_N, \\
    \label{eq:marginals2}
    \int_{A\times B} \ud p_2 &= \int_{A\times B\times\R^{3(N-2)}} \ud p_N,
\end{align}
for any subsets $A, B\subset \R^3$. 

We employ the usual notation $p_N\mapsto p_1$ and $p_N\mapsto p_2$ for the validity of Eq.~\eqref{eq:onebodydens} respectively \eqref{eq:twobodydens}.
\begin{definition} ($N$-density representability) Let $N\ge 2$. A probability density (or probability measure) $p_2$ on $\R^{6}$ is called \emph{$N$-density-representable} if there exists a symmetric 
probability density (or probability measure) $p_N$ on $\R^{3N}$ such that
$p_N\mapsto p_2$. 
\end{definition} 
\noindent
{\bf Examples} 1) It is clear that $p_2$ is 2-density representable if and only if it is symmetric. \\[1mm]
2) Any statistically independent measure $p_2(\vect{x},\vect{y}) = p_1(\vect{x}) p_1(\vect{y})$ is $N$-density representable for all $N$, since it is represented by the 
$N$-body probability measure $p_1(\vect{x}_1)\cdots p_1(\vect{x}_N)$. \\[1mm]
3) The totally anticorrelated probability measure 
\begin{equation}
	\label{eq:anticorr}
	p_2(\vect{x},\vect{y})=\frac12 \Bigl(\delta(\vect{x}-A) \delta(\vect{y}-B) + \delta(\vect{x}-B) \delta(\vect{y}-A)\Bigr),
\end{equation}
$A, B \in \R^3$, $A \neq B$ is 2-density representable, but not $3$-density representable. That is to say, even though it is a symmetric probability measure on $\R^6$, there does not exist
any symmetric probability measure $p_3$ on $\R^9$ such that $\int p_3(\vect{x},\vect{y},\vect{z}) \ud\vect{z} = p_2(\vect{x},\vect{y})$. 
The reason is explained in Section~\ref{sec:ModelProblem}.
\\[1mm]
4) The previous example can be turned into a smooth one
(see Figure~\ref{fig:smoothanti}). The smooth pair density
\begin{equation}
	\label{eq:smoothanticorr}
	p_2(\vect{x},\vect{y}) = \frac12 \Bigl(\varphi(\vect{x}-A) \varphi(\vect{y}-B) + \varphi(\vect{x}-B) \varphi(\vect{y}-A)\Bigr),
\end{equation}
with $\varphi$ any nonnegative function on $\R^3$ with $\int \varphi = 1$ and $\varphi(z)=0$ when $\abs{z} > \abs{A-B}/2$, is not $3$-density representable, as we will show in Section~\ref{sec:NecessaryConditions}.

\begin{figure}[ht!]
\centering
\includegraphics[height=0.3\textheight]{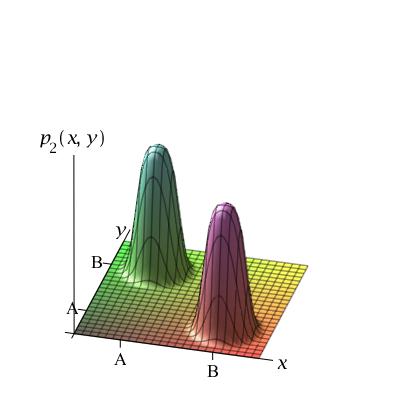}
\caption{Pair densities which are not $3$-density-representable, such as the one depicted here 
(Eq.~\eqref{eq:smoothanticorr}), can be quite innocent looking. For further discussion of this
example see Section~\ref{sec:NecessaryConditions}.}
\label{fig:smoothanti}
\end{figure} 

The above definition immediately implies the following theorem. 
\begin{theorem} \label{T:representability}
Let $N> M \ge 2$. If a probability density (or probability measure) $p_2$ on $\R^6$ is $N$-density-representable, then it is also $M$-density-representable. 
\end{theorem}
In other words, $N$-density representability becomes a more and more stringent condition as $N$ increases.
\begin{proof}
If $p_N$ is a symmetric $N$-body density which represents $p_2$, then 
\begin{multline*}
    \tilde{p}(\vect{x}_1,\dots,\vect{x}_M)\\
    = \int_{\R^{3(N-M)}} p_N(\vect{x}_1,\dots,\vect{x}_M,\dots,\vect{x}_N) \ud \vect{x}_{M+1} \cdots \ud \vect{x}_N
\end{multline*}
is a symmetric $M$-body density which also represents $p_2$.
\end{proof}
With the help of the concept of density representability, we can exploit the fact that the Coulomb potential $V_{\ee}$ in \eqref{eq:alternative3} only involves pair interactions to reformulate the many-body optimal transport definition \eqref{eq:alternative3}
of $V_{\ee}^{\SCE}$ as a standard (two-body) optimal transport problem with a constraint. This result does not depend on the Coulombic form of the interaction potential $v_{\ee}$.
\begin{theorem} (SCE energy via density representability) \label{T:SCE} For any given single-particle density $\rho$ of an $N$-electron system,
\begin{equation}
\label{eq:alternative4}
V_{\ee}^{\SCE}[\rho] = \min_{\substack{p_2\mapsto \rho/N \\ p_2 \text{ $N$-density-rep.}}} \binom{N}{2} \int_{\R^6} v_{\ee} \, p_2,
\end{equation}
where $v_{\ee}$ is any pair potential which is symmetric (i.e., $v(\vect{x},\vect{y}) = v(\vect{y},\vect{x})$), and the minimization is over probability densities $p_2$ on $\R^6$. 
\end{theorem}
Here the appearance of the normalization constants $N$ and $\binom{N}{2}$ is due to the fact that $\rho$ integrates to $N$, not $1$, whereas $p_2$ integrates
to $1$, not $\binom{N}{2}$. 
\begin{proof}
The proof is similar to the famous proof of Levy of the Hohenberg-Kohn theorem. For any symmetric $N$-point density $\rho_N$ with $\rho_N \mapsto \rho_2$, we clearly have
\begin{equation} \label{eq:Veeviap2}
    \int_{\R^{3N}} V_{\ee} \, \rho_N = \binom{N}{2} \int_{\R^6} v_{\ee} \, p_2;
\end{equation}
that is to say the electron-electron energy only depends on the two-body reduced density of $\rho_N$. 
We can therefore usefully partition the minimization in \eqref{eq:alternative3} into a double minimization, first over $\rho_N$ subject to fixed $p_2$, then over $p_2$:
\begin{eqnarray*}
   \min_{\rho_N\mapsto \rho/N} \int V_{\ee} \rho_N & = & \min_{\substack{p_2\mapsto \rho/N, \\ p_2 \text{ $N$-density-rep.}}} \min_{\rho_N\mapsto p_2} \int V_{\ee}\rho_N \\
    & = & \min_{\substack{p_2\mapsto \rho/N, \\ p_2 \text{ $N$-density-rep.}}} \binom{N}{2} \int v_{\ee} p_2,
\end{eqnarray*}
the last equality being due to Eq.~\eqref{eq:Veeviap2}. 
\end{proof}
The formula for $V_{\ee}^{\SCE}$ in Theorem~\ref{T:SCE} together with the necessary conditions for $N$-density representability in Theorem~\ref{T:representability} immediately suggests a natural hierarchy of approximations. For any given single-particle density $\rho$ of an $N$-electron system, let us define
\begin{multline} \label{eq:rdmodels}
   V_{\ee}^{\SCE,k}[\rho] = \min_{\substack{p_2\mapsto \rho/N, \\ p_2 \text{ $k$-density-rep.}}} \binom{N}{2} \int_{\R^6} v_{\ee} p_2, \\ k = 2, 3, \dots
\end{multline}
That is, we replace the requirement that $p_2$ be $N$-representable by the weaker requirement that it be $k$-representable for some $k\le N$. This enlarges the set of admissible $p_2$'s in the minimization, leading to the following chain of inequalities
\begin{equation}
	\label{eq:rdineq}
	\begin{array}{ccccccc} V_{\ee}^{\SCE,2}[\rho] & \le & \dots & V_{\ee}^{\SCE,k}[\rho] & \le & \dots  & V_{\ee}^{\SCE,N}[\rho] \\
	\parallel & & & & & & \parallel \\
	\min\limits_{p_2\mapsto\rho/N} \int v_{\ee}\, p_2 & & & & & & V_{\ee}^{\SCE}[\rho].
	\end{array}
\end{equation} 
We call $V_{\ee}^{\SCE,k}$ the order-$k$ approximation of the SCE energy.
The lowest-order approximation $V_{\ee}^{\SCE,2}$ corresponds to solving a classical (two-body) optimal transport problem with 
Coulomb cost (yielding the functional introduced in Ref.~\onlinecite{CFK1} which we called $F^{OT}[\rho]$), whereas the order-N approximation $V_{\ee}^{\SCE,N}$ recovers the exact SCE energy. Physically, the intermediate functionals $V_{\ee}^{\SCE,k}$ can be thought of as reduced models for the energy of strongly correlated electrons which take into account $k$-body correlations. 
\\[2mm]
The unknown in the order-$k$ approximation is a $k$-body density, so the computational cost increases 
steeply with $k$; e.g., discretizing each copy of $\R^3$ by $K$ gridpoints leads to a $K^k$-point discretization 
for $\R^{3k}$. The practical value of our reduced models therefore depends strongly on whether low-order approximations are already capable of capturing the main part of the full SCE energy. A tentative answer is
that they are, as we will document in the next two sections.
\\[2mm]
Theoretically the parameter $k$ in \eqref{eq:rdmodels} can also be chosen bigger than $N$, a particularly interesting question being what happens when
$k\to\infty$. In the present paper, we will only answer this question for the model densities \eqref{eq:twodiracs} below. A general discussion will appear elsewhere.
\\[2mm]
Finally, we remark that density representability of a pair density is obviously a necessary condition for the familiar (wavefunction) representability of any two-body density matrix
which gives rise to this pair density. Analyzing the relationship between this necessary condition and common representability conditions from density matrix theory such as the $P$, $Q$ and $G$ conditions\cite{Coleman2000,Mazziotti2007} lies beyond the scope of this paper.

\section{Model problem: $N$ particles occupying 2 sites}
\label{sec:ModelProblem}

In this section we analyze a model system in which the particle positions are restricted to 2 sites, to gain basic insights into what it means for a pair density 
to be $k$-density representable and into how the resulting functionals \eqref{eq:rdmodels} depend on $k$. The single-particle density of such a system has the form
\begin{equation} \label{eq:twodiracs}
	\frac{\rho}{N} = (1-t)\,\delta_{A} + t\,\delta_{B}, \quad 0 \le t \le 1,
\end{equation}
where $N$ is the number of particles and $A$ and $B$ are two different points in $\R^3$. This model density, while of course very simplistic, can be regarded as a toy model for the electron density of a diatomic system in the regime when the interatomic distance is much larger than the atomic radii. If we don't want to allow fractional occupation numbers of the sites, $t$ would be restricted to integer multiples of $1/N$, but since this makes little difference to the analysis, we might as well allow real occupation numbers.
\\[2mm]
Our first goal is to compute explicitly the set of $N$-representable $2$-point probability measures for 
our 2-site system. The $N$-point probability measures on $\R^{3N}$ whose single-particle density has
the form \eqref{eq:twodiracs} for some $t$ are the measures of form
\begin{equation}
	\label{eq:Npointmeasures}
	p_N = \sum_{I=(i_1,\dots,i_N)\in\{A,B\}^N} \alpha_I \, \delta_{i_1}\otimes \cdots \otimes\delta_{i_N}
\end{equation}
with $\alpha_I\ge 0$, $\sum_I\alpha_I = 1$, and correspond to the probability measures on the discrete $2$-site, $N$-particle state space 
$\{A,B\}^N$. We use the following notation for the different sets of probability measures of interest:
$\calP(\{A,B\}^N)$ denotes the set of probability measures on $\{A,B\}^N$, i.e., all measures of form 
\eqref{eq:Npointmeasures}; 
$\calP^{\text{sym}}(\{A,B\}^N)$ is the set of such measures which are \emph{symmetric}, i.e., $\alpha_{(i_1,\dots,i_N)}$ is a symmetric function of its arguments $(i_1,\dots,i_N)$; 
and $\calP^{N\text{-rep}}(\{A,B\}^2)$ stands for the set of $N$-density-representable probability measures 
on the two-body state space $\{A,B\}^2$, i.e., those probability measures on $\{A,B\}$ which arise as
marginals \eqref{eq:twobodydens} of some $p_N\in\calP^{\text{sym}}(\{A,B\}^N)$. In particular, the 
$2$-density-representable probability measures are those measures of form
\begin{multline}
	\label{eq:2rep1}
    p_2 = \alpha_{AA} \delta_A\otimes\delta_A + \alpha_{BB} \delta_B\otimes\delta_B\\
    + \alpha_{AB}\delta_{A}\otimes\delta_B + \alpha_{BA}\delta_{B}
    \otimes \delta_A 
\end{multline}
which satisfy the trivial conditions of nonnegativity, normalization, and symmetry,
\begin{equation}
	\label{eq:2rep2}
    \alpha_{ij}\ge 0 \text{ for all } i,j, \quad \sum_{i,j} \alpha_{ij} = 1, \quad \alpha_{AB}=\alpha_{BA}.
\end{equation}
It is clear from the explicit representation \eqref{eq:Npointmeasures} that $\calP(\{A,B\}^N)$ is the convex hull of its 
extreme points $\delta_{i_1}\otimes \cdots \otimes \delta_{i_N}$, $i_1,\dots,i_N\in\{A,B\}$. The set of symmetric 
$N$-point probability densities satisfies $\calP^{\text{sym}}(\{A,B\}^N) = S_N \calP(\{A,B\}^N)$, where $S_N$ is the symmetrizer
\begin{equation*}
(S_N\,p_N)(A_1\times\cdots\times A_N) = \frac{1}{N!} \sum_{\sigma} p_N(A_{\sigma(1)}\times \cdots \times A_{\sigma(N)})
\end{equation*}
and the sum runs over all permutations. It follows that $\calP^{\text{sym}}(\{A,B\}^N)$ is the convex hull of the elements
$S_N\delta_{i_1}\otimes \cdots \otimes \delta_{i_N}$, and that $\calP^{N\text{-rep}}(\{A,B\}^2)$ is the convex hull of 
their two-point densities,
\begin{multline}
	\label{eq:ExtDens}
	\calP^{N\text{-rep}}(\{A,B\}) = \text{convex hull of the measures} \\
	\Bigl\{p_2^{(S_N \delta_{i_1} \otimes \cdots \otimes \delta_{i_N})}, \, i_1,\dots,i_N\in\{A,B\}\Bigr\},
\end{multline}
where here and below, $p_2^{(p_N)}$ denotes the two-particle density of $p_N$.
To compute these two-point densities, we use an averaging formula which can be shown by an elementary computation:
first symmetrizing and then taking the two-point density is the same as taking the average over all two-point
densities,   
\begin{equation}
	\label{eq:MarOfSym}
	p_2^{(S_N\,p_N)} = \frac{1}{\binom{N}{2}} \sum_{1\le i<j\le N} \int p_N \, \ud\widehat{\vect{x}}_{ij},
\end{equation}
where $\widehat{\vect{x}}_{ij}$ denotes the list of coordinates $\vect{x}_1,\dots,\vect{x}_N$ with $\vect{x}_i$ and $\vect{x}_j$ omitted, and $p_N$
is any $N$-point probability measure. Now consider
a measure of form $p_N=\delta_{i_1}\otimes\cdots\otimes\delta_{i_N}$, and let $K = \sharp\{i_j \, | \, i_j=B\}$, i.e.~the occupation number of site $B$. Note $0\le K\le N$. By the averaging formula \eqref{eq:MarOfSym}, and using the abbreviated notation $\delta_i^K = \otimes_{i=1}^K \delta_i$, 
\begin{multline}
\label{eq:MarFinal}
p_2^{(S_N\delta_{i_1}\otimes\cdots\otimes\delta_{i_N})} 
= p_2^{(S_N\delta_A^{N-K} \otimes \delta_B^{K})} \\ 
= \frac{1}{\binom{N}{2}} \Bigl[ \binom{N-K}{2} \delta_A\otimes\delta_A + \binom{K}{2} \delta_B\otimes\delta_B \\
+ \frac{K(N-K)}{2}\bigl(\delta_A\otimes\delta_B + \delta_B\otimes\delta_A\bigr)\Bigr].
\end{multline}
Note that the resulting two-point marginal does not depend on the $i_j$, but only on the occupation number
$K\in\{0,\dots,N\}$. Equations \eqref{eq:ExtDens}, \eqref{eq:MarFinal} give the following final result:
\begin{theorem} \label{T:Nrep} The set of $N$-representable
2-point measures, $\calP^{N\text{-rep}}(\{A,B\})$, is the convex hull of the $K\!+\! 1$ measures given by the right hand side
of \eqref{eq:MarFinal}, where $K$ runs from $0$ to $N$.
\end{theorem}
This set is plotted in Figure~\ref{fig:Nrep}, for different values of $N$. 
\begin{figure}
\centering
\includegraphics[width=\columnwidth]{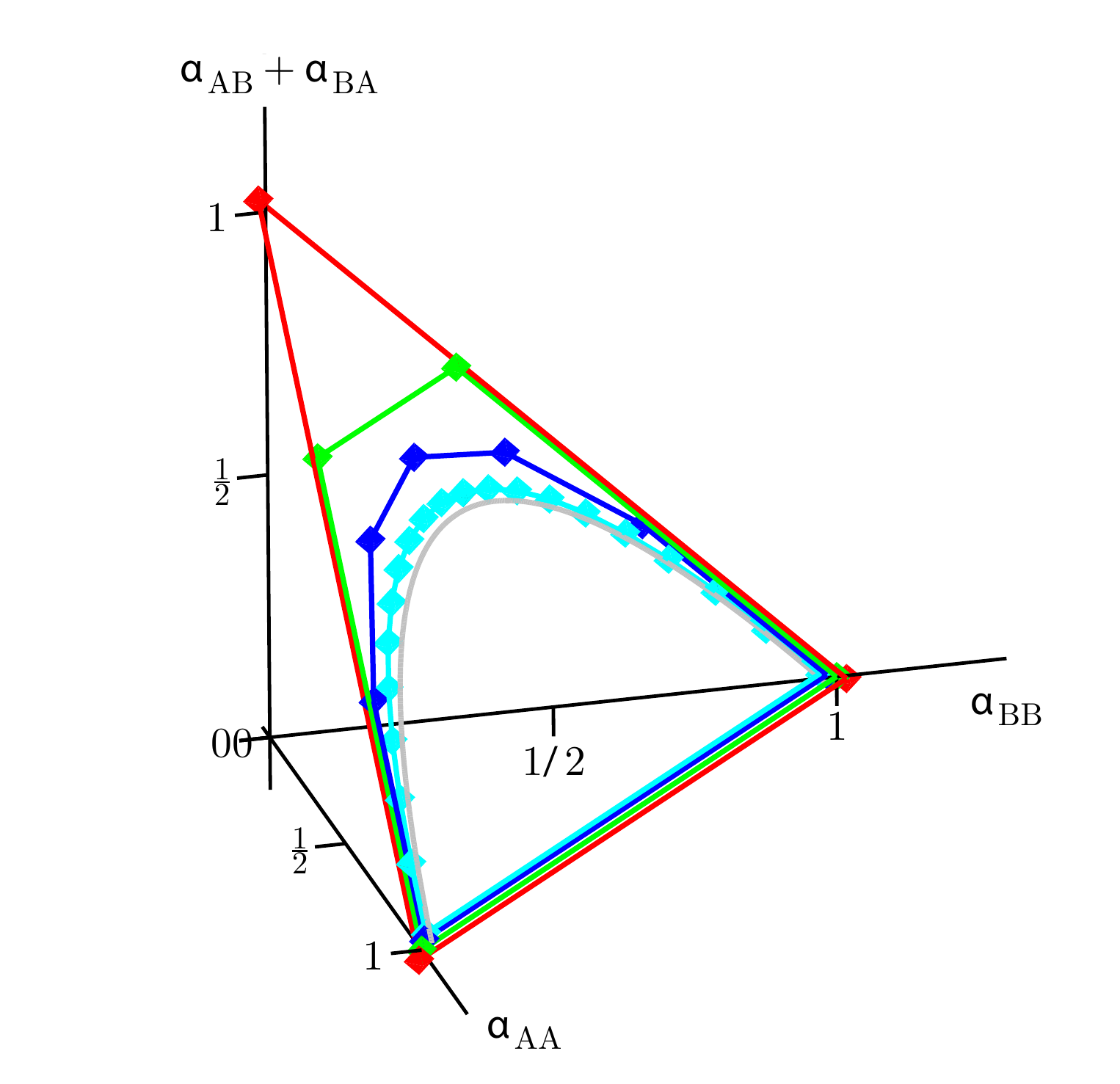}
\caption{The set of $N$-representable pair densities of form $\alpha_{AA}\delta_A\otimes\delta_A + 
\alpha_{AB}\delta_A\otimes\delta_B + \alpha_{BA}\delta_B\otimes\delta_A + \alpha_{BB}\delta_B\otimes\delta_B$
for $N=2$ (red), $N=3$ (green), $N=6$ (blue), and $N=20$ (light blue). 
The coefficient $\alpha_{AB}+\alpha_{BA}$ on the vertical axis indicates the weight
of the anticorrelated contribution \eqref{eq:anticorr}; its maximum representable value decreases with $N$.   
When $N=2$, only the trivial conditions \eqref{eq:2rep2} are present. Remarkably, as $N$ gets large 
the upper boundary of the representable set approaches the curve given by the mean field densities, i.e.,
$p_2=p_1\otimes p_1$ for some $p_1$ (grey curve), see Eq.~\eqref{eq:MarViaMeanfield},\eqref{eq:MeanField}.}
\label{fig:Nrep}
\end{figure}

Next we show that, as suggested by Figure~\ref{fig:Nrep}, when $N$ gets large the $N\! +\! 1$ extremal pair densities in 
\eqref{eq:MarFinal} approach a certain very interesting continuous curve. To see this, let us re-write formula \eqref{eq:MarFinal} in 
terms of the normalized occupation number $t=K/N\in[0,1]$ instead of $K$, and separate the coefficients into
$N$-independent and lower order terms. An elementary calculation shows that
\begin{equation*}
	\frac{K-1}{N-1} = t - \frac{1-t}{N-1}, \quad \frac{(N-K)-1}{N-1} = (1-t) - \frac{t}{N-1},
\end{equation*}
and, abbreviating $\delta_i \otimes \delta_j$ by $\delta_{i j}$,
\begin{multline}
	\label{eq:MarViaMeanfield}
	p_2^{(S_N\delta_A^{N-K} \otimes \delta_B^{K})} \\
	= \underbrace{(1-t)^2 \delta_{A A} + t^2 \delta_{B B} + t(1-t)(\delta_{A B} + \delta_{B A})}_{= p_2^{(\infty)}} \\
	+ \frac{t(1-t)}{N-1} \left(- \delta_{A A} - \delta_{B B} + \delta_{A B} + \delta_{B A}\right). 
\end{multline}
But the first term is precisely the mean field approximation to the pair density of the state 
$S_N\delta_A^{N-K}\delta_B^K$ obtained from its single-particle density 
\begin{equation}
	p_1^{(S_N\delta_A^{N-K}\delta_B^K)} = (1-t)\delta_A + t\delta_B,
\end{equation}
namely 
\begin{equation}
	\label{eq:MeanField}
	p_2^{(\infty)} = \bigl((1-t)\delta_A + t\delta_B\bigr) \otimes \bigl((1-t)\delta_A + t\delta_B\bigr).
\end{equation}
The second term in \eqref{eq:MarViaMeanfield} is a correlation correction which depletes the ``ionic'' terms $\delta_A\otimes\delta_A$ and 
$\delta_B\otimes\delta_B$ in favour of the ``anticorrelated'' 
terms $\delta_A\otimes\delta_B$ and $\delta_B\otimes\delta_A$. This correction is large for small $N$ (and 
even completely removes the ionic terms when $N=2$ and $t=1/2$), but vanishes in the limit 
$N\to\infty$ at fixed occupation number $t$. 

In particular, we have established the following
\begin{theorem} \label{T:MeanField}
A pair density of form \eqref{eq:2rep1} is $N$-representable for all $N$ if and only if it lies in the convex
hull of the mean field densities, or -- by inspection of Figure~\ref{fig:Nrep} -- 
if and only if it is a convex combination of \emph{two} mean field densities.  
\end{theorem}

The ``primal'' description of $N$-representable pair densities 
as the convex hull of explicit extreme points can be easily turned into an equivalent ``dual'' description
via inequalities. We only give the result in the cases $N=3$ and $N=\infty$.
\begin{corollary}
A pair density of form 
\eqref{eq:2rep1} is 3-representable if and only if it satisfies \eqref{eq:2rep2} and the linear inequality
\begin{equation}
	\label{3rep}
	\alpha_{AB}+\alpha_{BA} \le 2(\alpha_{AA}+\alpha_{BB}),
\end{equation}
and $N$-representable for all $N$ if and only if it satisfies \eqref{eq:2rep2} and the nonlinear inequality
\begin{equation}
	\label{eq:inftyrep}
	\alpha_{AB} + \alpha_{BA} \le 2 \sqrt{\alpha_{AA}\cdot \alpha_{BB}}. 
\end{equation}
\end{corollary}
To derive \eqref{eq:inftyrep}, one first shows that $\alpha_{AB}+\alpha_{BA} \le 2(\alpha_{AA}+\alpha_{AB})(\alpha_{BB}+\alpha_{BA})$. Thanks to Eq.~\eqref{eq:2rep2} this is a quadratic inequality for
$\alpha_{AB}+\alpha_{BA}$ and solving it yields \eqref{eq:inftyrep}.

The physical meaning of Eq.~\eqref{3rep} is that at most $2/3$ of the mass of $p_2$ can sit 
on the non-ionic configurations $(A,B)$ and $(B,A)$. The meaning of Eq.~\eqref{eq:inftyrep}
is that the total size of the non-ionic contributions cannot exceed its size
in the mean field pair density formed from its single-particle density. 
\\[2mm]
The above results can be used to determine the hierarchy of approximate $N$-particle functionals $V_{\ee}^{\SCE,k}$ introduced in \eqref{eq:rdmodels} on densities
of the form \eqref{eq:twodiracs}.
In fact, the exact Coulomb potential $v_{\ee}(\vect{x},\vect{y}) = \abs{\vect{x}-\vect{y}}^{-1}$ no longer makes sense in the context of these model densities since multiply occupied sites would lead
to infinite energy, so we replace it by an appropriately regularized interaction, with the property that
\begin{multline}
	\label{eq:RegInt}
	{v}_{\ee}(A,A)={v}_{\ee}(B,B) = U_{\text{diag}}\\
	> U_{AB} = {v}_{\ee}(A,B)={v}_{\ee}(B,A).
\end{multline}
Here $U_{\text{diag}}$ and $U_{AB}$ are effective parameters for same-site and different-site repulsion, and the
inequality $U_{\text{diag}}>U_{AB}$ preserves the repulsive effect that the interaction potential decreases 
with interparticle distance. Hence the two-point densities $p_2$ which compete in the 
variational definition \eqref{eq:rdmodels} prefer the different-site configurations $(\vect{x},\vect{y})=(A,B)$ and $(\vect{x},\vect{y})=(B,A)$
over the ionic configurations $(A,A)$ and $(B,B)$. 
Consequently the optimizing $p_2$'s with one-point density \eqref{eq:twodiracs} are those $k$-representable 2-point densities of form \eqref{eq:2rep1}
which have one-point density $\rho_t$ (this fixes their position in direction of the baseline in Figure~\ref{fig:Nrep}, because $t=\frac12(\alpha_{BB}-\alpha_{AA})$) which 
maximize the coefficient $\alpha_{AB}+\alpha_{BA}$, i.e., lie on the upper boundary of the representable set in Figure~\ref{fig:Nrep}. 
When $t$ is an integer multiple of $1/k$, i.e., $t=K/k$, $K=0,1,\dots,k$, the optimizing $p_2$ is thus precisely given by formula \eqref{eq:MarViaMeanfield} with $N$ replaced by $k$.
It follows that, denoting the right hand side of \eqref{eq:twodiracs} by $\rho_t$,
\begin{multline}
	\label{eq:rdenergies}
	{V}_{\ee}^{\SCE,k}[\rho_t] \\
	= \binom{N}{2} \Bigl( U_{\text{diag}} \cdot [t^2 + (1-t)^2] + U_{AB} \cdot 2t(1-t)\Bigr) \\
	- \binom{N}{2} (U_{\text{diag}}-U_{A B}) \frac{2 t(1-t)}{k-1},\\
	t = K/k, \quad K = 0, 1,\dots, k.
\end{multline}
For intermediate occupation numbers $t$ with $t_- =(K-1)/k < t < K/k =t_+$, $K=1,2,\dots,k$, the upper boundary
of the $k$-representable set is given by the linear interpolation between the $p_2$'s coming from $t_-$ 
and $t_+$, and hence
so is the resulting value of $\tilde{V}_{\ee}^{\SCE,k}$. The weights of the contributions from $t_\pm$ are the same as the interpolation weights for the single-particle density,
$\rho_t = (K-k t)\rho_{t_-} + (k t - (K-1))\rho_{t_+}$, and so 
\begin{multline}
	\label{eq:rdenergies2}
	{V}_{\ee}^{\SCE,k}[\rho_t] = (K-kt) {V}_{\ee}^{\SCE,k}[\rho_{(K-1)/k}] \\
	+ \bigr(kt-(K-1)\bigl) {V}_{\ee}^{\SCE,k}[\rho_{K/k}], \\
	\frac{K-1}{k} \le t \le \frac{K}{k}, \quad K = 1,\dots,k.
\end{multline}
The reduced SCE energies \eqref{eq:rdenergies}, \eqref{eq:rdenergies2} are plotted in Figure~\ref{fig:en}, for different values of $k$. 

\begin{figure}
\centering
\includegraphics[width=0.9\columnwidth]{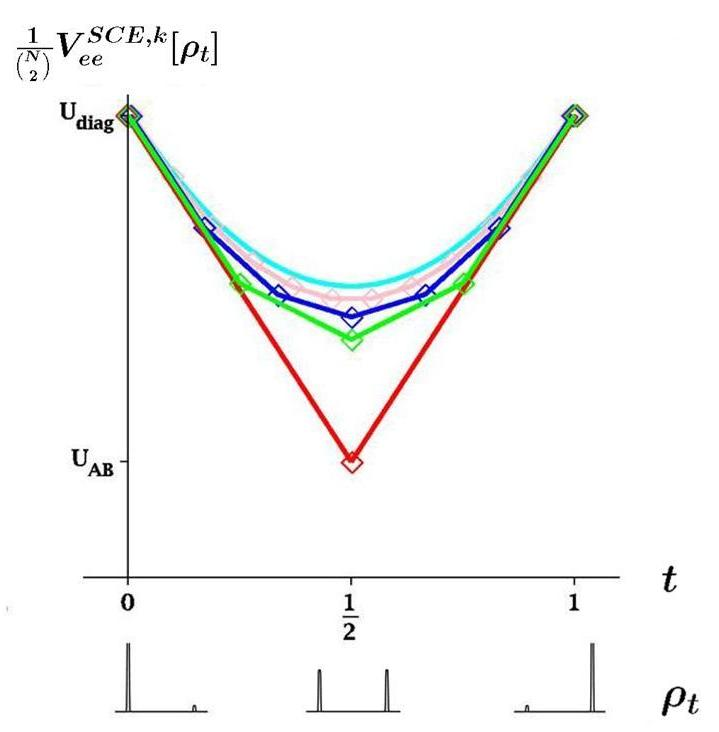}
\caption{Density-representability approximation of order $k$ to $V_{\ee}^{\SCE}$, on densities of form 
$\rho_t/N = (1-t)\delta_A + t\delta_B$ (various values of $k$). The Coulomb interaction has been replaced by the regularized interaction \eqref{eq:RegInt}. Red, green, blue, and pink corresponds to $k = 2, 4, 6, 11$. The piecewise linear structure is an exact feature of the results \eqref{eq:rdenergies}, \eqref{eq:rdenergies2}. The order-$k$
approximation equals the exact $V_{\ee}^{\SCE}$ for $k=N$ particles. Light blue curve: mean field energy (see text).}
\label{fig:en}
\end{figure}


Finally let us calculate and physically interpret the large-$k$ limit. It is clear from our explicit results that the limit is just given by the first part of the right hand side
of Eq.~\eqref{eq:rdenergies}. This part is nothing but a (Hartree-type) mean field energy,
\begin{equation}
\begin{split}
\label{hartree}
\lim_{k\to\infty}{V}_{\ee}^{\SCE,k}[\rho_t]
&= \binom{N}{2} \int {v}_{\ee} \frac{\rho_t}{N}\otimes\frac{\rho_t}{N}
= \left(1-\tfrac{1}{N}\right) J[\rho_t], \\
J[\rho] &= \frac12 \int {v}_{\ee}(\vect{x},\vect{y}) \rho(\vect{x}) \rho(\vect{y}) \ud \vect{x} \ud \vect{y}.
\end{split}
\end{equation}
Here the prefactor $1 - 1/N$ is a self-interaction correction, i.e., the approximation via density representability
of infinite order 
remembers that there are only $\binom{N}{2}$ interaction terms, not
$N^2/2$. In other words, remarkably, \emph{the infinite-order approximation to the SCE functional
is nothing but the self-interaction-corrected mean field energy}, even though mean field approximations 
played no role in the construction of the reduced SCE functionals \eqref{eq:rdmodels}.

After completing the above (elementary) analysis, we learned that results similar to, and in fact much more general than, Theorems~\ref{T:Nrep} and \ref{T:MeanField} 
are well known in probability theory, more precisely in the theory of ``exchangeable sequences'' of random variables\cite{Ald85, DF80}. This theory, which appears to be hitherto
disconnected from DFT (as well as wavefunction representability), entails a classification going back to de Finetti\cite{deFin69} of symmetric probability densities 
$p_{\infty}(\vect{x}_1,\vect{x}_2,\dots)$ in infinitely many variables.

\section{Necessary conditions for density-representability}
\label{sec:NecessaryConditions}

The results on two-state systems in the previous section immediately yield necessary conditions on $N$-density representability
for general pair densities $p_2$ on $\R^6$. To this end, let us introduce the following integrals of $p_2$ associated with
any partitioning of $\R^3$ into two disjoint subsets $\Omega_A$ and $\Omega_B$:
\begin{equation}
	\label{eq:integrals_p2}
	\alpha_{ij} = \int_{\Omega_i\times\Omega_j} p_2, \quad i,j\in\{A,B\}.
\end{equation}
\begin{theorem}
\label{T:NecRep}
Let $p_2$ be any pair density (or measure) on $\R^6$, normalized so that $\int p_2=1$. If $p_2$ is
$N$-density-representable, and $\Omega_A$, $\Omega_B$ is any partitioning of $\R^3$ into two subsets, then the associated $2$-site pair density
\begin{equation}
	\label{eq:assoc2site}
	\alpha_{AA}\delta_{AA} + \alpha_{AB}\delta_{AB} + \alpha_{BA}\delta_{BA} + \alpha_{BB}\delta_{BB},
\end{equation}
with $\alpha_{ij}$ as in \eqref{eq:integrals_p2}, is also $N$-representable, that is to say it belongs to the set $\calP^{N\text{-rep}}(\{A,B\})$ computed explicitly in Theorem~\ref{T:Nrep} and depicted in Figure~\ref{fig:Nrep}.
\end{theorem}
\begin{proof}
If $p_N$ is an $N$-point density (or measure) on $\R^{3N}$ which represents $p_2$, then the associated 2-site, $N$-point density
\begin{multline}
     \sum_{i_1,\dots,i_N \in \{A,B\}} \alpha_{i_1 \cdots i_N} \, \delta_{i_1}\otimes\cdots\otimes\delta_{i_N}\\
     \text{with} \quad \alpha_{i_1\cdots i_N} = \int_{\Omega_{i_1}\times\cdots\times\Omega_{i_N}} p_N,
\end{multline}
represents the 2-site pair density \eqref{eq:assoc2site}.
\end{proof}

{\bf Example} The smooth anticorrelated pair density \eqref{eq:smoothanticorr} in Example 4 of Section~\ref{sec:NdensityRepr} (see Figure~\ref{fig:smoothanti}) is not 3-representable: choose $\Omega_A$, $\Omega_B$ to be
the half-spaces of $\R^3$ whose boundary bisects the line segment from $A$ to $B$, that is to say 
$\Omega_A=\{\vect{x}\in\R^3\, | \, \vect{x} \cdot (B-A) \le M \cdot (B-A)\}$, 
$\Omega_B=\{\vect{x}\in\R^3\, | \, \vect{x} \cdot (B-A) > M \cdot (B-A)\}$, where $M=(A+B)/2$. By construction, in this case $\alpha_{AA}=\alpha_{BB}=0$, $\alpha_{AB}=\alpha_{BA}=1$, so the associated 2-site pair density
\eqref{eq:assoc2site} is not $3$-density representable, as shown in the previous section (see Figure~\ref{fig:Nrep}).

\medskip

\noindent In fact, for $N=3$, the dual description \eqref{3rep} of the $N$-representable pair densities for the $2$-site system yields the necessary condition of Theorem~\ref{T:NecRep} directly in the form of the following inequality:
\begin{equation}
	\label{eq:nec3rep}
    \int_{\Omega_A\times\Omega_B} p_2 + \int_{\Omega_B\times\Omega_A}p_2 \le 2 \left(\int_{\Omega_A\times\Omega_A}p_2 + \int_{\Omega_B\times\Omega_B}p_2 \right). 
\end{equation} 
Physically, this condition says that the total probability to find a particle pair in the ``anticorrelated'' regions $\Omega_A\times\Omega_B$ and $\Omega_B\times\Omega_A$
can be at most twice as large as the probability to find a pair in the ``ionic'' regions $\Omega_A\times\Omega_A$ and $\Omega_B\times\Omega_B$.

\section{Behaviour of the reduced models on ab initio densities for small atoms}
\label{sec:AtomsFixedDensity}

We now investigate the effect of the hierarchy of representability conditions (Section~\ref{sec:NdensityRepr}) for the atoms He, Li and Be. The ab initio single electron density $\rho$ is obtained from a full configuration interaction (FCI) calculation with Slater-type orbitals (STOs) \cite{AsymptoticsCI2009}. The approximate interaction energy $V_{\ee}^{\SCE,k}$ of an $N$-electron system can be obtained by simulating a fictitious $k$-electron system: directly from Eq.~\eqref{eq:rdmodels}, we have
\begin{equation}
V_{\ee}^{\SCE,k}[\rho] = \frac{\binom{N}{2}}{\binom{k}{2}} V_{\ee}^{\SCE}\left[\tfrac{k}{N} \rho\right].
\end{equation}
The energy on the right hand side is obtained by the same method as in Ref.~\onlinecite{SGS2007}; in particular, the jump surfaces of the maps $\vect{T}_i$ in \eqref{eq:ansatz} are assumed to be concentric spheres.

Fig.~\ref{fig:reprVsce} compares the $k$-density representability approximation $V^{\SCE,k}_{\ee}[\rho]$ with $V^{\SCE}_{\ee}[\rho]$, with the ``exact'' value from the FCI calculation, and the Hartree term $J$. $V^{\SCE}_{\ee}$ underestimates the exact value, whereas $J$ overestimates it. The order-$k$ approximation is already reasonably close to $V^{\SCE}_{\ee}$. The corresponding numerical values are summarized in Table~\ref{tab:reprVsce}.

\begin{figure}[!ht]
\centering
\includegraphics[width=\columnwidth]{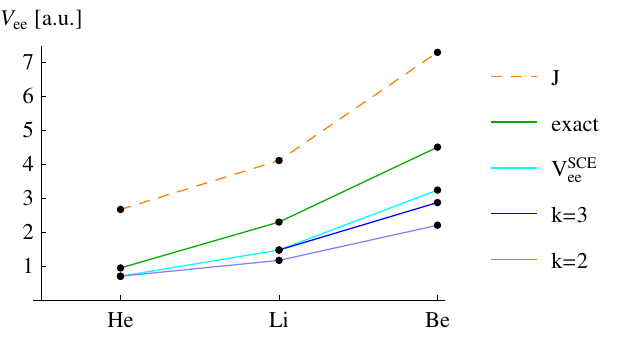}
\caption{$V^{\SCE,k}_{\ee}[\rho]$ obtained by the $k$-density representability approximation (blue) and the true $V^{\SCE}_{\ee}[\rho]$ (cyan). The green curve shows the ``exact'' $\langle\Psi | V_{\ee} | \Psi\rangle$ within an STO FCI ansatz space. The mean-field Hartree term $J$ (orange dashed) overestimates the exact value, as expected. See Tab.~\ref{tab:reprVsce} for the numerical values.}
\label{fig:reprVsce}
\end{figure}

\begin{table}
\centering
\begin{tabular}{c|ccc}
& \text{He} & \text{Li} & \text{Be} \\
\hline
$k = 2$ & $\mathbf{0.711906}$ & 1.17881 & 2.21361 \\
$k = 3$ &    & $\mathbf{1.48426}$ & 2.88229 \\
$k = 4$ &    &    & $\mathbf{3.24853}$ \\
\hline
\text{exact} & 0.954988 & 2.30755 & 4.51366 \\
\hline
\text{J} & 2.67842 & 4.11866 & 7.30589 \\
\end{tabular}
\caption{The calculated values of $V_{\ee}^{\SCE,k}$ using the $k$-density representability approximation (see also Fig.~\ref{fig:reprVsce}). Each diagonal entry is the true $V_{\ee}^{\SCE}$. An ``exact'' (FCI) value and the Hartree term $J$ are also shown, for comparison.}
\label{tab:reprVsce}
\end{table}

%
%

\section{A self-consistent Kohn-Sham computation comparing exact and reduced SCE}
\label{sec:SCF_KS}

The SCE formalism has the potential to become an important ingredient in the design of exchange-correlation functionals for strongly correlated electron systems. Thus we investigate the SCE approach in the context of a Kohn-Sham self-consistent field calculation for atoms with the total energy functional
\begin{equation} \label{eq:DF}
E[\rho] = T_{\mathrm{KS}}[\rho] + V_{\ee}^{\SCE,k}[\rho] + \int v_{\text{ext}}(\vect{x}) \, \rho(\vect{x}) \ud\vect{x}.
\end{equation}
Here $T_{KS}$ is the Kohn-Sham kinetic energy functional and $v_{\text{ext}}$ is the external nuclear potential. Previous self-consistent field calculations with the SCE
functional were carried out in Ref.~\onlinecite{MaletGiorgi2012, ML2013} for a 1D quantum wire, where in the weak confinement regime the SCE functional becomes asymptotically
exact. For 3D atomic systems considered here, replacing $J + E_{\mathrm{xc}}$ by $V_{\ee}^{\SCE,k}$
(or even by the exact SCE functional $V_{\ee}^{\SCE}$) presumably does not yield physically accurate results due to the missing influence of kinetic energy on $\rho_2$, but our calculations illustrate the effect of the $k$-density approximation.

The Kohn-Sham equations define a nonlinear eigenvalue problem
\begin{align}
H[\rho] \psi_i &= \varepsilon_i \psi_i \label{eq:KSE} \\
\rho(\vect{x}) &= \sum_{i=1}^N \abs{\psi_i(\vect{x})}^2, \quad \int \psi_i(\vect{x})^* \, \psi_j(\vect{x}) \ud\vect{x} = \delta_{ij},
\end{align}
where the Hamiltonian $H[\rho]$ itself depends on the density $\rho$. For an atom with nuclear charge $Z$, and the density functional \eqref{eq:DF} with $k=N$ (exact SCE), the single-particle Hamiltonian (in atomic units) reads
\begin{equation}
\label{eq:KSHamiltonian}
H[\rho] = -\frac{1}{2} \Delta - \frac{Z}{\abs{\vect{x}}} + u[\rho].
\end{equation}
The term $-\frac{1}{2} \Delta - {Z}/\abs{\vect{x}}$ is the hydrogen-like single-particle Hamiltonian, and $u[\rho]$ is the Kantorovich potential, i.e., the maximizer of
\eqref{eq:Dual}, which enters because formally, $(\delta V_{\ee}^{\SCE}/\delta\rho)[\rho] = u[\rho]$. Note that changing the potential in \eqref{eq:KSHamiltonian} 
by an additive constant would not change the Kohn-Sham orbitals, but choosing precisely $u[\rho]$ has the virtue that the Kohn-Sham eigenvalues sum to the system energy $E[\rho]$, as is easily inferred from \eqref{eq:KSE}, \eqref{eq:Dual}. 

The Kantorovich potential agrees up to an additive constant with the effective SCE potential $v[\rho]$ constructed in Ref.~\onlinecite{SGS2007}. The latter can be defined by\cite{SGS2007}
\begin{equation}
\label{eq:vGrad}
\nabla v[\rho](\vect{x}) = - \sum_{i=2}^N \frac{\vect{x} - \vect{T}_i(\vect{x})}{\abs{\vect{x} - \vect{T}_i(\vect{x})}^3}, \quad \lim_{\abs{\vect{x}} \to \infty} v[\rho](\vect{x}) = 0.
\end{equation}
The $\vect{T}_i$ are the transport maps in Eq.~\eqref{eq:ansatz}, which determine the positions of the remaining electrons given the position of the first electron, and solely depend on the density $\rho$. The additive constant is now easily obtained from \eqref{eq:Dual}: $u[\rho] = v[\rho]+C$ with
\begin{equation}
\label{eq:constant}
C = \int \frac{\rho(\vect{x})}{N} \sum_{i<j} \frac{1}{\abs{\vect{T}_i(\vect{x}) - \vect{T}_j(\vect{x})}} \ud\vect{x} - \int \rho(\vect{x}) v[\rho](\vect{x}) \ud\vect{x}.
\end{equation}
Following Ref.~\onlinecite{SGS2007} we assume on physical grounds that
\begin{equation}
\label{eq:vAsymptotic}
v[\rho](\vect{x}) \sim \frac{N-1}{\abs{\vect{x}}} \quad\text{as}\quad \abs{\vect{x}} \to \infty,
\end{equation}
even though we do not know of a mathematical proof. For charge-neutral atoms with $N = Z$, the Hamiltonian~\eqref{eq:KSHamiltonian} can be re-written as
\begin{equation}
\label{eq:KSHamiltonianZ1}
H[\rho] = -\frac{1}{2} - \frac{1}{\abs{\vect{x}}} + C + \left(v[\rho](\vect{x}) - \frac{N - 1}{\abs{\vect{x}}}\right),
\end{equation}
such that the last term is expected to decay faster than $1/\abs{\vect{x}}$ due to the asymptotic relation~\eqref{eq:vAsymptotic}.


\begin{figure*}[!ht]
\centering
\subfloat[density]{
\label{fig:KS_SCE_He_density}
\includegraphics[width=0.48\textwidth]{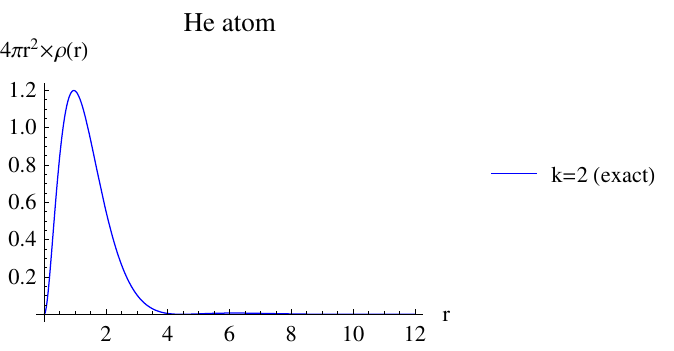}}
\hspace{0.02\textwidth}
\subfloat[potential]{
\label{fig:KS_SCE_He_v}
\includegraphics[width=0.48\textwidth]{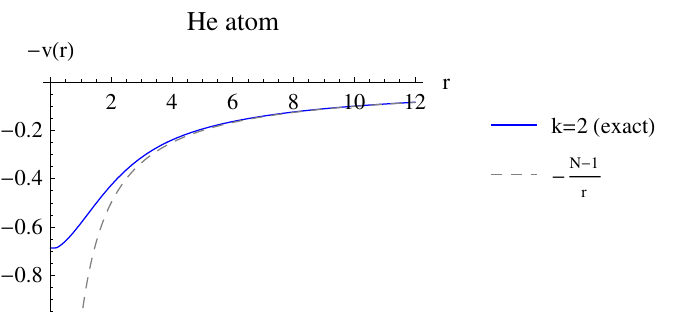}} \\
\subfloat[density]{
\label{fig:KS_SCE_Li_density}
\includegraphics[width=0.48\textwidth]{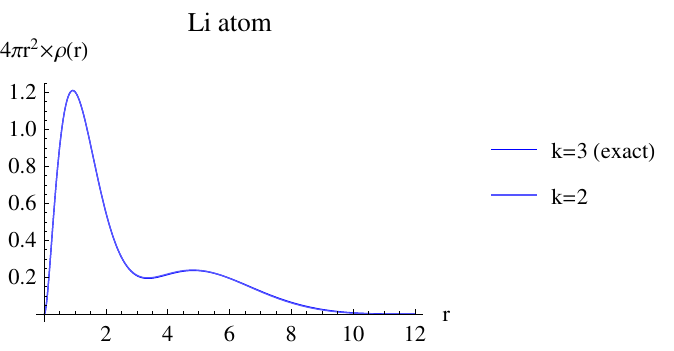}}
\hspace{0.02\textwidth}
\subfloat[potential]{
\label{fig:KS_SCE_Li_v}
\includegraphics[width=0.48\textwidth]{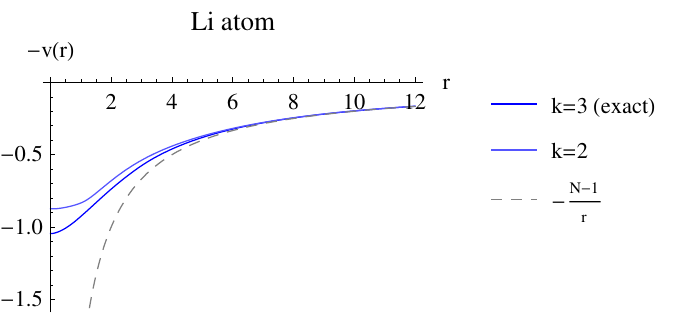}} \\
\subfloat[density]{
\label{fig:KS_SCE_Be_density}
\includegraphics[width=0.48\textwidth]{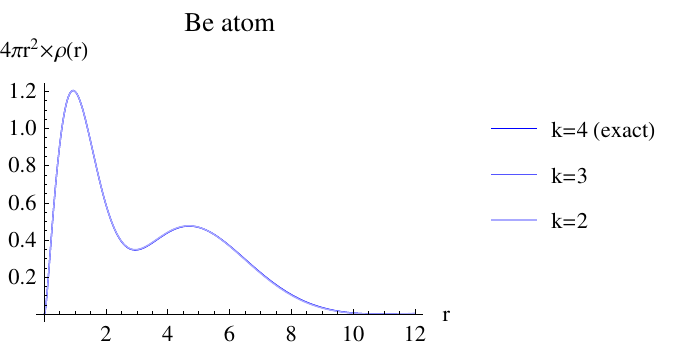}}
\hspace{0.02\textwidth}
\subfloat[potential]{
\label{fig:KS_SCE_Be_v}
\includegraphics[width=0.48\textwidth]{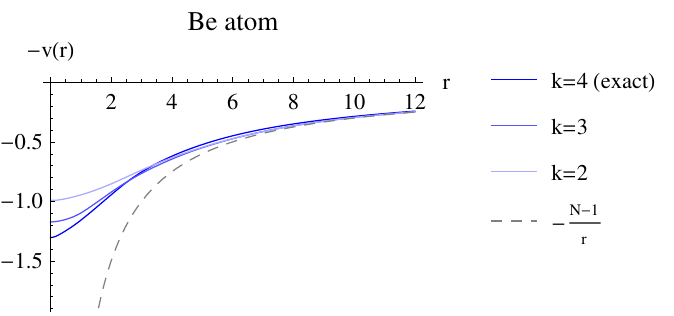}}
\caption{(a, c, e) Radial part of the self-consistent density (blue) for the helium, lithium and beryllium atom with the SCE exchange-correlation functional $E_{\mathrm{xc}} = V_{\ee}^{\SCE,k} - J$. The lighter blue curves correspond to the $k$-density representability approximation, and are visually indiscernible from the (exact) $k = N$ case. (b, d, f) The SCE (alias shifted Kantorovich) potential $v[\rho]$ (blue) corresponding to the self-consistent density on the left, rescaled by $\frac{N-1}{k-1}$. The gray dashed line shows the asymptotic expansion $-\frac{N-1}{\abs{\vect{x}}}$ in Eq.~\eqref{eq:vAsymptotic}.}
\label{fig:KS_SCE_Be}
\end{figure*}


Fig.~\ref{fig:KS_SCE_Be} shows the self-consistent densities and corresponding SCE potential $v$ of helium, lithium and beryllium, for the exact SCE potential as well as its $k$-density approximation (obtained from simulating a fictitious $k$-electron system, see Section~\ref{sec:AtomsFixedDensity}). All densities are normalized to $N$. For each $k$, $v[\rho]$ is rescaled by $\frac{N-1}{k-1}$ to match the asymptotic expansion~\eqref{eq:vAsymptotic}. Table~\ref{tab:KS_SCE} summarizes the numerical Kohn-Sham energy (sum of Kohn-Sham eigenvalues, with doubly occupied orbitals due to spin). Note that the $k$-density representability approximation of the energy is below the true value and increases with $k$.




\begin{table}
\centering
\begin{tabular}{c|ccc}
& \text{He} & \text{Li} & \text{Be} \\
\hline
$k = 2$ & -2.74058 & -5.24439 & -8.12061 \\
$k = 3$ &          & -5.10739 & -7.90984 \\
$k = 4$ &          &          & -7.79889 \\
\end{tabular}
\caption{Kohn-Sham energy (sum of eigenvalues) obtained by a self-consistent field iteration with the SCE Kantorovich potential and Hamiltonian in Eq.~\eqref{eq:KSHamiltonianZ1}.}
\label{tab:KS_SCE}
\end{table}

\section{Example of a minimizing $N$-point density not of SCE form}
\label{sec:MongeCounterexample}

Since much of the theoretical and numerical work on the minimization
problem \eqref{eq:alternative3} relies on the (plausible but nontrivial)  
ansatz \eqref{eq:ansatz}, it is of interest to understand its precise
status with respect to minimization over arbitrary $N$-point probability
measures.  
 
For $N=2$ the ansatz is known to be exact, in the sense that the
minimizing $2$-point probability measure is unique and of the form
\eqref{eq:ansatz} (Ref.~\onlinecite{CFK1}, following
earlier work\cite{GM} in the optimal transportation
literature on pair interactions $v_{\ee}$ which increase with
interparticle distance).
 
For $N>2$, physical arguments\cite{SGS2007} suggest that there
should always exist a minimizing $N$-point probability measure of this
form, so in particular restricting the
minimization in \eqref{eq:alternative3} to $\rho_N$'s of form
\eqref{eq:ansatz} should always give the correct value of the functional
$V_{\ee}^{\SCE}[\rho]$, but there is no rigorous proof of this conjecture.  
 
Finally, there is the question whether for $N>2$ the ansatz
\eqref{eq:ansatz} yields \emph{all} solutions. Here we are not aware of any
convincing arguments (be they physical or mathematical), one way or the
other. The following counterexample demonstrates that for $N>2$,
\eqref{eq:ansatz} does \emph{not} yield all solutions if the Coulomb
interaction is replaced by a negative harmonic oscillator interaction.
This is a new and somewhat surprising effect which only appears when
$N>2$; for $N=2$ the negative harmonic interaction  
was already considered\cite{SGS2007, BPG2012} in connection with the SCE functional and the two interactions were shown to behave in
exactly the same way.
 
Our counterexample does not imply that the ansatz \eqref{eq:ansatz} does
not capture all minimizers of the true Coulombic SCE problem
\eqref{eq:alternative3}, but it
means that if it does, this must be because of some special Coulombic
features.
 
The counterexample is best discussed in the context of recent work in
the optimal transport literature on $N$-body optimal transport problems
in $\R^{d N}$ with a general nonnegative interaction potential or ``cost
function'' $V_{\ee}(\vect{x}_1,\dots,\vect{x}_N)$,
\begin{equation*}
\min_{\rho_N\mapsto\rho} \int_{\R^{d N}}
V_{\ee} \, \rho_N,
\end{equation*}
where the minimization is over $N$-point probability measures on $\R^{dN}$. 
As shown in Ref.~\onlinecite{P}, minimizers have to concentrate on subsets whose
dimension is bounded in terms of the \emph{signatures} (the number of
positive, negative and zero eigenvalues) of certain symmetric matrices
derived from the mixed second order partial derivatives of $V_{\ee}$. Let $G$ be the off diagonal part of the Hessian of $V_{\ee}$. 
More explicitly, if
\begin{equation*}
D^2_{\vect{x}_i \vect{x}_j} V_{\ee} = \Bigg(\frac{\partial^2 V_{\ee}}{\partial x^{\alpha}_i \partial
x^{\beta}_j }\Bigg)_{\alpha \beta}
\end{equation*}
denotes the $d \times d$ matrix of mixed
second order partials with respect
to $\vect{x}_i\in\R^d$ and $\vect{x}_j\in\R^d$, we have
\begin{equation}
\label{eq:G}
G =
\begin{bmatrix}
0 & D^2_{\vect{x}_1\vect{x}_2}V_{\ee} & \dots &D^2_{\vect{x}_1\vect{x}_m}V_{\ee} \\
D^2_{\vect{x}_2\vect{x}_1}V_{\ee} & 0 & \dots &D^2_{\vect{x}_2\vect{x}_m}V_{\ee} \\
\vdots &\vdots &0 &\vdots \\
D^2_{\vect{x}_m\vect{x}_1}V_{\ee} & D^2_{\vect{x}_m\vect{x}_2}V_{\ee} & \dots & 0
\end{bmatrix}.
\end{equation}
Note that $G$ is a block matrix; each entry in the preceding formula
denotes a $d \times d$ block.  Now, as a symmetric $Nd \times Nd$
matrix, $G$ has $N d$ real eigenvalues, counted with multiplicities.  Let
$\lambda_+,\lambda_- $ and $\lambda_0$ denote, respectively, the number
of positive, negative and zero eigenvalues of $G$ at some point
$x = (\vect{x}_1,\vect{x}_2,\dots \vect{x}_N) \in \mathbb{R}^{Nd}$; note that $\lambda_+
+\lambda_- +\lambda_0 = N d$. Then, near $x$, Theorem 2.3 in Ref.~\onlinecite{P}
implies that the \emph{support} of minimizers (the subset on which they are nonzero) is contained in a subset of dimension $\lambda_0+\lambda_-$.
 
For the Coulomb interaction $\sum_{i<j} \abs{\vect{x}_i-\vect{x}_j}^{-1}$, a
straightforward calculation implies that  
\begin{multline*}
D^2_{\vect{x}_i\vect{x}_j}V_{\ee} \\
= \frac{1}{\abs{\vect{x}_i-\vect{x}_j}^3} \left(I-\frac{3}{\abs{\vect{x}_i-\vect{x}_j}^2}(\vect{x}_i-\vect{x}_j)(\vect{x}_i-\vect{x}_j)^T\right),
\end{multline*}
where $I$ is the $d \times d$ identity matrix.  The signature of $G$,
however, may change depending on the point $\vect{x}$.  One can show that
(except at special points)  
$d \leq \lambda_0+\lambda_-\leq (N-1)d$, meaning that the dimension of
the support of the solution can be no more than $(N-1)d$. In particular,
for $N=2$,
this yields an alternative justification of the ansatz \eqref{eq:ansatz}.  
 
While the preceding result is only an upper bound on the dimension, it
is nevertheless a useful guideline for constructing high-dimensional
minimizers, since the $G$-matrix of the
cost must then necessarily have a large number of nonpositive eigenvalues.  
 
For ease of analysis of the $G$-matrix, consider now a cost of pair
potential form, Eq.~\eqref{eq:Vee}, with
$v_{\ee}$ symmetric and quadratic.
The $d\times d$ block $D^2_{\vect{x}_i\vect{x}_j}V_{\ee}$ is then independent of $i$, $j$,
and $\vect{x}$, and the signatures of $G$ can be computed explicitly. The
maximum number of nonpositive eigenvalues
occurs for the negative harmonic oscillator interaction  
\begin{equation}
	\label{eq:negharmosc}
    v_{\ee}(\vect{x}_i,\vect{x}_j) = -\abs{\vect{x}_i-\vect{x}_j}^2.
\end{equation}
In this case each $D^2_{\vect{x}_i\vect{x}_j}V_{\ee} = 2I$, and
\begin{equation*}
G=
\begin{bmatrix}
0 & 2I & \dots &2I \\
2I & 0 & \dots &2I \\
\vdots &\vdots &0 &\vdots \\
2I & 2I & \dots &0
\end{bmatrix}.
\end{equation*}
Now for any $v\in\R^d$, $(v,v,\dots ,v)$ is an eigenvector with
eigenvalue $2(N-1)$, while the vectors  
\begin{equation*}
(v,-v,0, \dots, 0), \, (v,0,-v,0,\dots,0), \dots, (v,0,\dots,0,-v)
\end{equation*}
are eigenvectors with eigenvalue $-2$.  This implies that $\lambda_{-} =
(N-1)d$, $\lambda_{+} = d$ and $\lambda_0 = 0$.  Therefore, minimizers have
at most $(N-1)d$-dimensional support. We now show that this bound is
sharp for this cost function; that is, there actually are minimizers
which are strictly positive on $(N-1)d$-dimensional sets.

\medskip

\noindent {\bf Example} Replace the Coulomb interaction with the negative harmonic oscillator
interaction \eqref{eq:Vee}, \eqref{eq:negharmosc}.  
Let $\tilde{\rho}_N$ be any symmetric measure on $\R^{3N}$ which is
concentrated on the $3(N-1)$-dimensional surface $\{\vect{x}_1+\vect{x}_2+\dots +\vect{x}_N
=0\}$. Then this measure is optimal for the corresponding single
particle density $\rho(\vect{x}_1) = N\int_{\mathbb{R}^{3(N-1)}} \tilde{\rho}_N
(\vect{x}_1,\vect{x}_2,\dots,\vect{x}_N) \ud \vect{x}_2 \ud \vect{x}_3,\dots \ud \vect{x}_N$.

\medskip

\noindent To see why, note that by a simple computation
\begin{multline*}
	V_{\ee} = -\sum_{i<j} \abs{\vect{x}_i-\vect{x}_j}^2\\
	= \frac12 \abs{\vect{x}_1 + \dots + \vect{x}_N}^2 - \frac{N}{2} \sum_{i=1}^N \abs{\vect{x}_i}^2.
\end{multline*}
Hence for any $\rho_N$ with one-body density $\rho$
\begin{equation*}
  \int V_{\ee}\rho_N = \int \abs{\vect{x}_1+\dots+\vect{x}_N}^2 \rho_N + \frac12 \int \abs{\vect{x}_1}^2
\rho(\vect{x}_1) \, \ud\vect{x}_1.
\end{equation*}
The first term is minimized if and only if $\rho_N$ is zero outside the
surface $\vect{x}_1+\dots+\vect{x}_N=0$, and the second term only depends on the one-body
density $\rho$. Since $\tilde{\rho}_N$
vanishes outside this surface, it is a minimizer.

\medskip

The above example is in fact a special case of a result in Ref.~\onlinecite{P}. The
interested reader is encouraged to consult \onlinecite{P} for further results
on the dimension of the support  
of optimizers.

\section{Conclusions and Outlook}
\label{sec:Conclusions}

We have reformulated the strongly correlated limit of density functional theory via ``$N$-density representability'', i.e., the requirement that the pair density
comes from a symmetric $N$-point probability measure. This formulation gives rise to a natural hierarchy of approximate models, in which one
relaxes this requirement to the existence of a representing symmetric $k$-point density with $k<N$. In this paper we have presented a computational method
for the approximate models which is akin to a wavefunction method, in that the representing $k$-point density is resolved. One of the numerical findings we did not anticipate
is the extreme robustness of self-consistent Kohn-Sham densities with respect to the $k$-density approximation. 

For low $k$, a promising route towards extending
our methods to spherically asymmetric systems is the direct computation of the Kantorovich dual potential \cite{ML2013}. 
In the future, if a more direct understanding
of the main constraints on $\rho_2$ implied by $k$-representability can be obtained, one could also envision
a dual approach akin to reduced density matrix methods in which one would solve a constrained linear programming problem for the pair density.

Finally, another interesting issue raised by this work is to clarify the somewhat surprising connection between the SCE formalism and the mean field approximation suggested by our study of the two-site system in Section~\ref{sec:ModelProblem}.


\end{document}